\documentclass[journal,onecolumn]{IEEEtran}

\usepackage{amsmath,amssymb,amsfonts,amsthm}
\usepackage{inputenc}
\usepackage{enumerate}
\usepackage{graphicx}
\usepackage{multicol}
\usepackage{multirow}
\newtheorem{thm}{Theorem}
\newtheorem{prop}{Proposition}
\newtheorem{lem}{Lemma}
\newtheorem{cor}{Corollary}

\newtheorem{exam}{Example}

\def\0{{\mathbf 0}}

\newcommand{\F}{\mathbb{F}}
\newcommand{\Z}{\mathbb{Z}}

\begin{document}

%
%
%

\title{On linear codes with one-dimensional Euclidean hull and their applications to EAQECCs}

\author{
Lin Sok\thanks{ School of Mathematical Sciences, Anhui University, Hefei, Anhui, 230601, { \tt soklin\_heng@yahoo.com}}
}

\date{}
\maketitle
\begin{abstract}
The Euclidean hull of a linear code $C$ is the intersection of $C$ with its Euclidean dual $C^\perp$.  The hull with low dimensions gets much interest due to its crucial role in determining the complexity of algorithms for computing the automorphism group of a linear code and checking permutation equivalence of two linear codes. The Euclidean hull of a linear code has been applied to the so-called entanglement-assisted quantum error-correcting codes (EAQECCs) via classical error-correcting codes. In this paper, we consider linear codes with one-dimensional Euclidean hull from algebraic geometry codes. Several classes of optimal linear codes with one-dimensional Euclidean hull are constructed. Some new EAQECCs are presented. 
\end{abstract}

\section{Introduction}

 The hull of a linear code was studied by Assmus {\em et al.} \cite{AssKey} to classify finite projective planes. In \cite{Sendrier97}, Sendrier determined the expected dimension of the hull of a random $[n,k]$ code when $n$ and $k$ go to infinity.  The average dimension of the hulls of cyclic codes was given by Skersys \cite{Ske}. Sangwisut {\em et al. }\cite{SangJitLingUdom} gave enumerations of cyclic codes and negacyclic codes of length $n$ with hulls of a given dimension.
 
The hull of linear codes with low dimensions gets much interest due to its crucial role in determining the complexity of algorithms for computing the automorphism group of a linear code \cite{Leon82} and checking permutation equivalence of two linear codes \cite{Leon91,Sendrier00}. 
The Euclidean hull of a linear code $C$ is defined to be the intersection of $C$ with its Euclidean dual $C^\perp$. A linear code with zero-dimension hull is called linear complementary dual (LCD) code. In their paper \cite{CarMesTanQiPel18}, the authors gave a method to construct any $q$-ary LCD code over $\F_q$ ($q>3$)  with parameters $[n,k,d]$ from a linear code with the same parameters. Related works on LCD codes, can be found in \cite{CarGui,CarGunOzbOzkSol,CarMesTanQiPel18,CarMesTanQi19,CarMesTanQi18-2,CarMesTanQi19-2,ChenLiu,Jin,JinBee,LiDingLi,LiLiDingLiu,Massey, MesTanQi,ShiYueYan,YanLiuLiYang}. The works on linear codes with one-dimension hull are referred to  \cite{CarLiMes,LiZeng,QCM,Sok1D}. Linear codes with $k$-dimension Euclidean hull are self-orthogonal codes with parameter $[n,k]$ and are self-dual codes if $n=2k$. We refer to \cite{FangFu,Gue,GraGul,JinXin,Sok,TongWang,Yan} for the work on families of optimal Euclidean self-dual codes.

Another motivation of studying the hull of linear codes comes from the so-called entanglement-assisted quantum error-correcting codes (EAQECCs) firstly introduced Bowen \cite{Bow} and followed by Brun \emph{et al.} \cite{BDH06}. In \cite{BDH06}, they showed that if pre-shared entanglement between the encoder and decoder is available, the EAQECCs can be constructed via classical linear codes without self-orthogonality as in the case of quantum stabilizer codes \cite{AK01}. Wilde {\em et al. }\cite{WilBru} proposed two methods to construct EAQECCs from classical codes, which were known as Euclidean construction method and the Hermitian construction method. These methods have recently been generalized by Galindo {\em et al. } \cite{GalHerMatRua}.

MDS codes form an optimal family of classical codes. Due to their largest error correcting capability for given length and dimension. MDS codes are of great interest in both theory and practice. Reed-Solomon codes are known to be MDS. However, MDS linear codes exist in a very restrict condition on their lengths as the famous MDS conjecture states: for every linear
$[n, k, n -k + 1]$ MDS code over $\F_q,$ if $1 < k < q,$ then $n \le q + 1,$ except when $q$ is even and $k = 3 $ or $k = q -1,$ in which cases $n \le q + 2.$  It is well-known that an EAQECC is MDS if and only if the corresponding classical linear code is MDS. Luo \emph{et al.} \cite{LCC} presented several classes of GRS and extended GRS codes with Euclidean hulls of arbitrary dimensions and constructed some families of $q$-ary MDS EAQECCs with length $n \leq q+1$. In \cite{GGJT18}, Guenda \emph{et al.} investigated the $\ell$-intersection pair of linear codes, where they completely determined the $q$-ary MDS EAQECCs of length $n \leq q+1$ for all possible parameters.

Algebraic geometry codes were discovered by Goppa, where they were also called geometric Goppa codes. Goppa showed in his paper \cite{Goppa} how to construct linear codes from algebraic curves over a finite field. A remarkable property of these codes is that their parameters can be computed via the degree of a divisor, which allows a nice description of the
code. Despite a strongly theoretical construction, algebraic geometry (AG) codes have asymptotically good parameters, and it was the first time that linear codes improved the so-called Gilbert-Vasharmov bound. 
It is well-known that AG codes of genus zero are MDS codes. On the contrary to the MDS case, almost MDS and near MDS codes exist more frequently.  A near MDS $[n,k]_q$ code corresponds to an $(n,k)$-arc in the finite geometry $PG(k-1,q)$. In \cite{de,DodLan}, using tools from finite geometry, for $q$ and $k$ fixed, the authors gave upper bounds on the length $n$ of an almost MDS and near MDS code over $\F_q$, respectively, and they also proved existence and non-existence of such codes over small finite fields. 

Recently, Sok \cite{Sok1D} has constructed several families of MDS linear codes with one-dimensional Euclidean hull from AG codes of genus zero.

In this work, we consider the constructions of linear codes with one-dimensional hull. We use tools from algebraic function fields to study such codes. Sufficient conditions for a code to have one-dimensional hull are given, and we explicitly construct many families of linear codes with one-dimensional Euclidean hull from algebraic curves.

The paper is organized as follows: Section \ref{section:pre} gives preliminaries and background on algebraic geometry (AG) codes. Section \ref{section:construction} provides a first characterization of an AG code to have one-dimensional hull and some methods to construct linear codes with one-dimensional hull. Section \ref{section:construction2} gives a second characterization and some constructions. Section \ref{section:application} gives an application to EAQECCs. We give a concluding remark in Section \ref{section:conclusion}.

\section{Preliminaries}\label{section:pre}

Let $\F_q$ be the finite field with $q$ elements. A 
linear code of length $n$, dimension $k$  and minimum distance $d$ over ${{\mathbb F}_q}$ is denoted as $q$-ary $[n,k,d]$ code. If $C$ is an $[n,k,d]$ code, then from the Singleton bound, its minimum distance is bounded above by
$$d\le n-k+1.$$
A code meeting the above bound is called {\em Maximum Distance Separable} ({MDS}). A code is called {\em almost} MDS if its minimum distance is one unit less than the MDS case. 
A code is called {\it optimal} if it has the highest possible minimum distance for its length and dimension.

The {\em Euclidean inner product} of ${\bf{x}}=(x_1,
\dots, x_n)$ and ${\bf{y}}=(y_1, \dots, y_n)$ in ${\mathbb F}_q^n$ is
${\bf{x}}*{\bf{y}}=\sum_{i=1}^n x_i y_i$. The {\em dual} of $C$,
denoted by $C^{\perp}$, is the set of vectors orthogonal to every
codeword of $C$ under the Euclidean inner product. The hull of a linear code $
C$ is $$hull(C):=C\cap C^\perp .$$ A linear code
$C$ is called $s$-$dim$ hull if $\dim (hull(C))=s.$ With this definition, a linear complementary dual (LCD) code is a $0$-$dim$ hull code and a self-dual code of length $n$ is a $\frac{n}{2}$-$dim$ hull code.

From Section \ref{section:pre} to Section \ref{section:construction}, we only consider the dual of a  linear code $C$ with respect to Euclidean inner product.

For a linear code $C\subset \F_q^n$ and ${\bf a}=(a_1,\hdots,a_n)\in (\F_q^*)^n$, we define

$$
{\bf a}\cdot C:=\{{\bf a} \cdot {\bf c}| {\bf c}\in C\}.
$$
It can be easily checked that ${\bf a}\cdot C$ is a linear code if and only if $C$ is a linear code. Moreover,  the codes $C$ and ${\bf a}\cdot C$ have the same dimension, minimum Hamming distance, and weight distribution.

We refer to Stichtenoth \cite{Stich} for undefined terms related to algebraic function fields. 

Let ${\cal X}$ be a smooth projective curve of genus $g$ over $\F_q.$
The field of rational functions of ${\cal X}$ is denoted by $\F_q({\cal X}).$ Function fields of
algebraic curves over a finite field can be characterized as
finite separable extensions of $\F_q(x)$. We identify points on the curve ${\cal X}$ with places of the
function field $\F_q({\cal X}).$ A point on ${\cal X}$ is called rational if all of
its coordinates belong to $\F_q.$ Rational points can be identified
with places of degree one. We denote the set of $\F_q$-rational
points of ${\cal X}$ by  ${\cal X}(\F_q)$.

A divisor $G$ on the curve ${\cal X}$ is a formal sum $\sum\limits_{P\in {\cal X}}n_PP$ with only finitely many nonzeros $n_P\in \Z$.
The support of $G$ is defined as $supp(G):=\{P|n_P\not=0\}$. The degree of $G$ is defined by $\deg(G):=\sum\limits_{P\in {\cal X}}n_P\deg(P)$. 
For two divisors $G=\sum\limits_{P\in {\cal X}}n_PP$ and $H=\sum\limits_{P\in {\cal X}}m_PP$,  we say that $G\ge H$ if $n_P\ge m_P$ for all places $P\in {\cal X}$, and define

$$G\wedge H:=\sum\limits_{P\in {\cal X}}\min (m_P,n_P)P,~G\vee H:=\sum\limits_{P\in {\cal X}}\max (m_P,n_P)P$$

For a nonzero rational function $f$ on the curve $\cal X$, we define the ``principal" divisor of $f$ as
$$(f):=\sum\limits_{P\in {\cal X}}v_P(f)P.$$ 

If $Z(f)$ and  $N(f)$ denotes the set of zeros and poles of $f$ respectively, we define the zero divisor and pole divisor of $f$, respectively by

$$
\begin{array}{c}
(f)_0:=\sum\limits_{P\in Z(f)}v_{P}(f)P,\\
(f)_\infty:=\sum\limits_{P\in N(f)}-v_{P}(f)P.\\
\end{array}
$$
Then, $(f)=(f)_0-(f)_\infty $, and it is well-known that the principal divisor $f$ has degree $0.$

We say that two divisors $G$ and $H$ on the curve $\cal X$ are equivalent if $G=H+(z)$ for some rational function $z\in \F_q({\cal X}).$

For a divisor $G$ on the curve $\cal X$, we define
$${\cal L}(G):=\{f\in \F_q({\cal X})\backslash \{0\}|(f)+G\ge 0\}\cup \{0\},$$ 
and 
$${\Omega}(G):=\{\omega\in \Omega\backslash \{0\}|(\omega)-G\ge 0\}\cup \{0\},$$
where $\Omega:=\{fdx|f\in \F_q({\cal X})\}$, the set of differential forms on $\cal X$. It is well-known that, for a differential form $\omega$ on $\cal X$, there exists a unique a rational function $f$ on $\cal X$ such that $$\omega=fdt,$$
where $t$ is a local uniformizing parameters. In this case, we define the divisor associated to $\omega$ by $$(\omega)=\sum\limits_{P\in {\cal X}}v_P(\omega)P, $$
where $v_P(\omega):=v_P(f).$

The dimension of ${\cal L}(G)$ is denoted by $\ell (G),$ and is determined by Riemann-Roch's theorem as follows.

\begin{thm}\cite[Theorem 1.5.15 (Riemann-Roch)]{Stich}\label{thm:Riemann-Roch} Let $W$ be a canonical divisor. Then, for each divisor $G$, the following holds:
$$\ell (G) = \deg G + 1 -g + i(G), \text{ where }i(G)=\ell(W-G),$$
and $g$ is the genus of the smooth algebraic curve.
\end{thm}
We call $i(G)$ the index of speciality of $G$. A divisor G
is called non-special if $i(G) = 0$ and otherwise it is called special.

For a particular divisor $G$, we can determine the dimension of the space ${\cal L}(G)$ as follows.
\begin{lem}\cite[Corollary 1.4.12]{Stich} \label{lem:dim-principal} Assume that a divisor $G$ has degree zero. Then, $G$ is principal
if and only if $\ell (G)=1.$
\end{lem}

Through out the paper, we let $D=P_1+\cdots+P_n$, called the rational divisor, where $(P_i)_{1\le i \le n}$ are places of degree one, and $G$ a divisor with $supp(D)\cap supp(G)=\emptyset$. Define the algebraic geometry code by
$$
C_{\cal L}(D,G):=\{(f(P_1),\hdots,f(P_n))|f\in {\cal L}(G)\},
$$
and the differential algebraic geometry code as
$$
C_{\Omega}(D,G):=\{(\text{Res}_{P_1}(\omega),\hdots,\text{Res}_{P_n}(\omega))|\omega\in {\Omega}(G-D)\},
$$
where $\text{Res}_{P}(\omega)$ denotes the residue of $\omega$ at point $P.$

The parameters of an algebraic geometry code $C_{\cal L}(D,G)$ is given as follows.
\begin{thm}\cite[Corollary 2.2.3]{Stich}\label{thm:distance} Assume that $2g -2 < deg(G) < n.$ Then, the code $C_{\cal L}(D,G)$  has parameters $[n,k,d]$ satisfying
\begin{equation}
k=\deg (G)-g+1\text{ and } d\ge n-\deg (G).
\label{eq:distance}
\end{equation}

\end{thm}

The dual of the algebraic geometry code $C_{\cal L}(D,G)$ can be described as follows.

\begin{lem}\cite[Theorem 2.2.8]{Stich}\label{lem:dual1} With above notation, the two codes $C_{\cal L}(D,G)$ and $C_{\Omega}(D,G)$ are dual to each other.
\end{lem}

Moreover, the differential code $C_{\Omega}(D,G)$ is determined as follows.
\begin{lem}\cite[Proposition 2.2.10] {Stich}\label{lem:dual2} With the above notation, assume that there exists a differential form $\omega $ satisfying
\begin{enumerate}
\item  $v_{P_i}(\omega)=-1,1\le i \le n$ and 
\item $\text{Res}_{P_i}(\omega)=\text{Res}_{P_j}(\omega)$ for $1\le i \le n.$ 
\end{enumerate}
Then, $C_{\Omega}(D,G)={\bf a}\cdot C_{\cal L}(D,D-G+(\omega))$ for some ${\bf a}\in ({\F^*_q}){^n}.$ 
\end{lem}

\section{First characterization and construction of $1$-$dim$ hull codes }\label{section:construction}
\subsection{First characterization }
In this section, we give a characterization of an algebraic geometry code to have one dimensional hull.
We begin with the following useful lemma.
\begin{lem}\label{lem:sufficient0} Let $(a_i)_{1\le i\le n}\in \F_q^*$ and $C$ be a linear code in $\F_q^n$. If $C^\perp={\bf e}\cdot C'$ with ${\bf e}=(a_1^2,\hdots,a_n^2)$ and $C\cap C'$ is a one-dimensional code, then $({\bf a}\cdot C)^\perp ={\bf a}\cdot C'$ and ${\bf a}\cdot C$ is a $1$-$dim$ hull code, where ${\bf a}=(a_1,\hdots ,a_n).$
\end{lem}

\begin{proof}
Since $C^\perp={\bf e}\cdot C'$, we have $\dim (C')+\dim(C)=n.$ For any ${\bf c}=(c_1,\hdots,c_n)\in C$, ${\bf c'}=(c'_1,\hdots,c'_n)\in C',$ we have $a{\bf c}*a{\bf c'}=a_1^2c_1c'_1+\cdots+a_n^2c_nc'_n=0$, and thus $({\bf a}\cdot C)^\perp={\bf a}\cdot C' .$ Now, we have $({\bf a}\cdot C)\cap ({\bf a}\cdot C)^\perp=({\bf a}\cdot C)\cap ({\bf a}\cdot C')={\bf a}\cdot(C\cap C')$, and thus the result follows.
\end{proof}
Now, an algebraic geometry code having one dimensional hull can be characterized as follows.
\begin{lem}\label{lem:sufficient}Let $E$ be a smooth algebraic curve over $\F_q$ with genus $g$  and $G, D = P_1 + P_2 + \cdots+ P_n$ be two divisors over $E$, where $0 < deg(G) < n$. Assume that there exists a differential form $\omega$ such that $(\omega) = G + H - D$ for some divisor $H$ with $Supp(G) \cap Supp(D) = Supp(H) \cap Supp(D) = \emptyset$. Assume further that
\begin{enumerate}
\item $G\wedge H$ is a principal divisor;
\item $\deg (G\vee H -D)= 0$;
\item $(\omega)$ is a principal divisor;
\item There is a vector ${\bf a}=(a_1,\hdots,a_n)\in (\F_q^*)^n$ with $Res_{P_i}(\omega)=a_i^2$.
\end{enumerate}
Then, ${\bf a}\cdot C_{\cal L}(D, G)$ is either a $1$-$dim$ or a $2$-$dim$  hull code with parameters $[n,\deg (G)-g+1, \ge n-\deg (G)]$, and its dual has parameters $[n,\deg (H)-g+1, \ge n-\deg (H)]$.
\end{lem}
\begin{proof} First, note that $G\vee H-D=(\omega)-G\wedge H $.

Let $c\in C_{\cal L}(D,G)\cap C_{\cal L}(D,H).$ Then we have that $c=(f(P_1),\hdots,f(P_n))$ $=(g(P_1),\hdots,g(P_n))$ for some $f\in {\cal L}(G)$ and $g\in {\cal L}(H).$ Set $e=f-g.$ If $e=0$, then $f=g\in {\cal L}(G)\cap {\cal L}(H)$. Thus $f=g\in {\cal L}(G\wedge H)$, and the fact that $\ell (G\wedge H)=1$ follows from Lemma \ref{lem:dim-principal}. If $e\not=0,$ then $e\in {\cal L}(G\vee H),$ and since $v_{P_i}(e)\ge 1$, it implies that $e\in {\cal L}(G\vee H-D).$ Since $(G\vee H-D)$ is principal and has degree 0, it follows that $\ell (G\vee H -D)=1$. Hence $f$ and $g$ are in one dimensional subspace. Now, take $C=C_{\cal L}(D,G)$ and $C'=C_{\cal L}(D,H)$, and then the fact that ${\bf a}\cdot C_{\cal L}(D,G)$ is either a $1$-$dim$ or a $2$-$dim$ hull code follows from Lemma  \ref{lem:dual2} and Lemma \ref{lem:sufficient0}. Finally, the dimension and the minimum distance follow from Thoerem \ref{thm:distance}.
\end{proof}
\begin{exam} \label{exam:1}Consider the elliptic curve over $\F_{5^2}$ defined by $$y^2z+yz^2=x^3+z^3.$$

The points $
O=(0:1:0),
   P_1= ( 4: 0 :1),
    P_2=( 4: 4 :1),
    P_3=( 1: 1 :1),
    P_4=( 1: 3 :1),
    P_5=( 2: w^8 :1),
    P_6=( 2: w^{16} :1),
    P_7=( 3: w^{13} :1),
    P_8=( 3:w^{17}:1)
$ are on the curve. 
Take $U=\{\alpha\in \F_{5^2}|\alpha^4=1\}$. Set $h(x)=\prod\limits_{\alpha\in U}(x-\alpha)$, $D=(h)_0$ and $\omega=\frac{z^4}{x^4-z^4}d(\frac{x}{z}).$ Then, the residue of $\omega$ at points $(P_i)_{1\le i\le 8}$ are in $\F_5$ and thus square elements in $\F_{5^2}.$
 Take $G=2O+P_{9}+P_{10},$ where  $P_{9}=(w^4:0:1),P_{10}=(w^4:4:1)$. We have that 
 $$
 \begin{array}{ll}
 \left(\frac{x-z}{z}\right)&=P_3+P_4-2O,\\
 \left(\frac{y-z}{x-w^4z}\right)&=P_3+(w^8:1:1)+(w^{16}:1:1)-O-P_9-P_{10},\\
  \left(\frac{x-y}{x-w^4z}\right)&=(4:1:1)+P_3-O-P_9-P_{10}.\\
 \end{array}
$$
Then, $1,\frac{x-z}{z},\frac{y-z}{x-w^4z},\frac{x-y}{x-w^4z}$ are in ${\cal L}(2O+P_{9}+P_{10})$ and thus form a basis for $C_{\cal L}(D,G)$. The generator matrix of $C_{\cal L}(D,G)$ is given by
$$
\left(
\begin{array}{cccccccc}
 1   & 1  &  1  &  1  &  1 &   1  &  1   & 1\\
   3  &  3 &   0   & 0   & 1   & 1  &  2   & 2\\
   w  &  0  &  0& w^{22} &w^{14} & w^2 &w^{17} &   w\\
   w & w^7  &  0 &w^{10} &   4 &w^{16}  &  w &w^{17}\\
\end{array}
\right).
$$

It can be checked from \cite{Mag} that the code $C_{\cal L}(D,G)$ has parameters $[8, 4, 4]$, and the code ${\bf a}\cdot C_{\cal L}(D,G) $, where ${\bf a}=(
    1,
    1,
    3,
    3,
    w^{21},
    w^{21},
    w^{15},
    w^{15}
)$, is a $1$-$dim$ hull code over $\F_{5^2}.$
\end{exam}
\subsection{First construction from elliptic curves}

First, let us recall some notation and definition about elliptic curves over finite fields. Let $E$ be an elliptic curve over $\F_q$ and $O$ be the point at infinity of $E(\overline {\F}_q)$, where $\overline {\F}_q$ is the algebraic closure of $\F_q$, and $E(\overline {\F}_q)$ is the set of all points on $E$. For a non-negative integer $r$, we define $E[r]:=\{P\in E(\overline{\F}_q):P\oplus\cdots \oplus P=O\}$

\begin{lem}\label{thm:01}Let $E$ be an elliptic curve over $\F_q$ and $G, D = P_1 + P_2 + \cdots+ P_n$ be divisors over $E$, where $0 < deg(G) < n$. Assume that there exists a differential form $\omega$ such that $(\omega) = G + H - D$ for some divisor $H$ with $Supp(G) \cap Supp(D) = Supp(H) \cap Supp(D) = \emptyset$. Assume further that
\begin{enumerate}

\item $G\wedge H$ is a principal divisor;
\item $\deg (G\vee H -D)=0$;
\item There is a vector ${\bf a}=(a_1,\hdots,a_n)\in (\F_q^*)^n$ with $Res_{P_i}(\omega)=a_i^2$.
\end{enumerate}
Then, ${\bf a}\cdot C_{\cal L}(D, G)$ is either a $1$-$dim$ or a $2$-$dim$ hull code with parameters $[n,\deg (G), \ge n-\deg (G)]$, and its dual has parameters $[n,n-\deg (G), \ge \deg (G)]$.
\end{lem}
\begin{proof} It is well-known that on an elliptic curve the Weil differential $(\omega) = G + H - D$ is principal. The rest follows from Lemma \ref{lem:sufficient}.
\end{proof}

\begin{thm}\label{thm:02}Let $D = P_1 + P_2 + \cdots+ P_n$ and $\overline{D}$ be such that $2O-\overline{D}$ is principal, $\overline{D}\notin E[r]$ with $1\le r<  n/4$ and $O, \overline{D}\notin supp(D)$. Set $G=2rO+r\overline{D}$. Assume that there exists a differential form $\omega$ such that $(\omega) = G + H - D$ for some divisor $H$ with $Supp(G) \cap Supp(D) = Supp(H) \cap Supp(D) = \emptyset$, and there exists a vector ${\bf a}=(a_1,\hdots,a_n)\in (\F_q^*)^n$ with $Res_{P_i}(\omega)=a_i^2$.
Then, ${\bf a}\cdot C_{\cal L}(D, G)$ is either a $1$-$dim$ or a $2$-$dim$ hull code with parameters $[n,4r, \ge n-4r]$, and its dual has parameters $[n,n-4r, \ge 4r]$.
\end{thm}

\begin{proof}Take $H=(n-2r)O-r\overline{D}$. Under the assumptions in the theorem, we have that $G\wedge H=2rO-r\overline{D}$ is principal and $G\vee H-D=(n-2r)O+r\overline{D}-D$ has degree zero. The result follows from Lemma \ref{thm:01}.
\end{proof}

We now derive explicit constructions of $1$-$dim$ hull codes from elliptic curves in even characteristic.

First, we will consider elliptic curves in Weierstrass form to construct one dimensional hull. Let $q = p^m$ and an elliptic curve be defined by the equation

\begin{equation}
{\cal E}_{a,b,c}:~y^2 + ay = x^3 + bx + c,
\label{eq:elliptic-curve}
\end{equation}
where $a, b, c \in\F_q .$ Let denote the number of rational points on ${\cal E}_{a,b,c}$ by ${N}_{a,b,c}$  Let $S$ be the set of $x$-components of the affine points of ${\cal E}_{a,b,c}$ over $\F_q $, that is, 
\begin{equation}
S_{a,b,c} := \{\alpha\in \F_q| \exists \beta \in \F_q\text{ such that } \beta^2+a\beta=\alpha^2+b\alpha+c\}.
\label{eq:set-elliptic}
\end{equation}

For $q=2^m,$ any $\alpha\in S_{1,b,c}$ gives exactly two points with $x$-component $\alpha$, and we denote these two points corresponding to $\alpha$ by $P^{(1)}_\alpha$ and $P^{(2)}_\alpha.$ Then, the set of all rational points of ${\cal E}_{1,b,c}$ over $\F_q$ is $\{P^{(1)}_\alpha|\alpha\in S_{1,b,c}\}\cup \{P^{(2)}_\alpha|\alpha\in S_{1,b,c}\}\cup \{O \}.$ 
The numbers ${N}_{1,b,c}$ of rational points of elliptic curves ${\cal E}_{1,b,c}$ over $\F_q$ are given in Table \ref{table:size-elliptic}.
\begin{table}
\caption{Numbers of rational points of elliptic curves}\label{table:size-elliptic}
\begin{center}
\begin{tabular}{c|c|c}
\text{Elliptic curve }${\cal E}_{1,b,c}$&$m$&${N}_{1,b,c}$\\
\hline
&$m\text{ odd }$&$q+1-2\sqrt{q}$\\
$y^2 + y = x^3$&$m\equiv 0 \pmod 4$&$q+1-2\sqrt{q}$\\
&$m\equiv 2 \pmod 4$&$q+1+2\sqrt{q}$\\
\hline
\multirow{2}{8em} 
{$y^2+y=x^3+x$}& $m\equiv 1,7 \pmod 8$ & $q+1+\sqrt{2q}$\\
&$m\equiv 3,5 \pmod 8$ & $q+1-\sqrt{2q}$\\
\hline
\multirow{2}{8em}{$y^2 + y = x^3 + x+1$}&$m\equiv 1,7 \pmod 8$&$q+1-\sqrt{2q}$\\
&$m\equiv 3,5 \pmod 8$&$q+1+\sqrt{2q}$\\
\hline
$y ^2 + y = x ^3 + bx ( T r_1^m ( b) = 1 )$&$m\text{ even }$&$q+1$\\
\hline
\multirow{2}{13em}{$y ^2 + y = x ^3 + c ~( T r_1^ m ( c ) = 1 )$}&$m\equiv 0 \pmod 4$&$q+1+2\sqrt{q}$\\
&$m\equiv 2 \pmod 4$&$q+1-2\sqrt{q}$\\
\end{tabular}
\end{center}
\end{table}

\begin{lem} Let $q$ be even and $s$ be a positive integer, $\{\alpha_1, \hdots , \alpha_s\}$ be a subset of $S_{1,b,c}$ with cardinality $s,$ and $D=\sum\limits_{i=1}^s(P_{\alpha_i}^{(1)}+P_{\alpha_i}^{(2)}).$ Let $h=\prod\limits_{i=1}^s(x+\alpha_i)$ and $\omega=\frac{dx}{h}$. Then, $(\omega)=2sO-D$ and 
$$
Res_{P_{\alpha_i}^{(1)}}(\omega)=Res_{P_{\alpha_i}^{(2)}}(\omega)=
\frac{1}{\prod\limits_{j=1,j\not= i}^s(\alpha_i+\alpha_j)},
$$
for any $i\in \{1,\hdots,s\}.$
\end{lem}
The above residue formula will be used for calculation in the following two corollaries. 

\begin{cor}\label{cor:explicit1}Let $q$ be even, $n=2s<{N}_{1,b,c}$. Assume that there exists a set $S=\{\alpha_0,\alpha_1,\hdots,\alpha_s\}\subset S_{1,b,c}$. Set
$D=P_{\alpha_1} + P_{\alpha_2}+\cdots+P_{\alpha_s}$ and $G=2rO+rP_{\alpha_0}$ with $1\le r<  \frac{n}{4}$ and $P_{\alpha_0}\notin E[r]$, where $P_{\alpha_i}=P^{(1)}_{\alpha_i}+P^{(2)}_{\alpha_i},0\le i\le s.$ Then for some ${\bf a}=(a_{11},a_{12},\hdots a_{s1},a_{s2})\in (\F_q^*)^n$, ${\bf a}\cdot C_{\cal L}(D, G)$ is either a $1$-$dim$ or a $2$-$dim$ hull code with parameters $[n,4r, \ge n-4r]$, and its dual has parameters $[n,n-4r, \ge 4r]$.
\end{cor}

\begin{proof} Take $\omega=\frac{dx}{h(x)}$, where $h(x)=\prod\limits_{i=1}^s(x-\alpha_i)$, and $H=D-G+(\omega)$. Obviously, the residue of the differential form 
$\omega$ at each point $P_{\alpha_i}$ is a nonzero square element in $\F_q$ for $q$ even. Take $a_{i1}^2=Res_{P_{\alpha_i}^{(1)}}(\omega), a_{i2}^2=Res_{P_{\alpha_i}^{(2)}}(\omega)$. Thus, the result follows from Theorem \ref{thm:02}.
\end{proof}

The following lemma \cite{JinXin} guarantees existence of some square elements in $\F_q$ for $q$ odd.
\begin{lem}\label{lem:JinXin}
For any given positive integer $m$, if $q \ge 4^m \times m^2$, then there exists a subset $S = \{\alpha_1,\hdots,\alpha_m\}$ of $\F_q$ such that $\alpha_j -
\alpha_i$ are nonzero square elements for all $1 \le i < j \le m.$
\end{lem}
\begin{cor}\label{cor:q-odd}
Let $q$ be an odd prime power with $q \equiv 1 \pmod 4 $, then for every $n\equiv 0 \pmod 4$  satisfying  $q \ge  4^{s+4} (s + 4)^2,$ there exist either $1$-$dim$ or $2$-$dim$ hull codes with parameters $[2s,4r,\ge 2s-4r]$ and $[2s, 2(s-2r), \ge 4r]$ for $1\le r<\frac{s}{2}$.
\end{cor}
\begin{proof}From Lemma \ref{lem:JinXin}, there exists $s+4$ distinct elements $a_1,a_2,a_3,\alpha_0,\alpha_1,\hdots, \alpha_s$ such that $\alpha_i-\alpha_j$, $\alpha_i-a_1$, $\alpha_i-a_2$, $\alpha_i-a_3$ are squares in $\F_q$ for $0\le i,j\le s$. Consider the elliptic curve defined by
\begin{equation}
{\cal E}: y^2=(x-a_1)(x-a_2)(x-a_3)
\end{equation}
Set $Q_1=(a_1,0),Q_2=(a_2,0),Q_3=(a_3,0)$. Note that if $(\alpha,\beta)$ is on ${\cal E}$, then so is $(\alpha,-\beta)$. For $0\le i\le s$, set $P_{\alpha_i}^{(1)}=(\alpha_i,\beta_i)$ and $P_{\alpha_i}^{(2)}=(\alpha_i,-\beta_i)$ to be points on ${\cal E}$. Put $A=\{1\le i\le s:\beta_i \text{ is a square in }\F_q\}$ and $B=\{1\le i\le s:\beta_i \text{ is not a square in }\F_q\}$. Then, $|A|+|B|=n$ and thus one of the two sets has at least $s$ elements.

Assume that $|A|\ge s$, say $A$ contains $\{1,2,\hdots,s\}$. Take $D=\sum\limits_{i=1}^{s}\left(P_{\alpha_i}^{(1)}+P_{\alpha_i}^{(2)}\right)$, ${\overline D}=P_{\alpha_0}^{(1)}+P_{\alpha_0}^{(2)}$, $G=2rO+r\overline D$ and $$\omega=\frac{ydx}{\prod\limits_{i=1}^s(x-\alpha_i)}.$$
Then, we get

$$Res_{P_{\alpha_i}^{(1)}}(\omega)=\frac{\beta_i}{\prod\limits_{j=1,j\not= i}^s(\alpha_i-\alpha_j)},Res_{P_{\alpha_i}^{(2)}}(\omega)=\frac{-\beta_i}{\prod\limits_{j=1,j\not= i}^s(\alpha_i-\alpha_j)}.$$
Since $q\equiv 1 \pmod 4$, $-1$ is a square in $\F_q$, and thus $Res_{P_{\alpha_i}^{(j)}}$ is a square in $\F_q$ for $1\le i\le s$ and $1\le j\le 2.$
Since $(dx)=(y),$ we get that $(\omega)=(2Q_1+2Q_2+2Q_3-6O)+(nO-D)$. Take $H=D-G+(\omega)$ and $a_{ij}^2=Res_{P_{\alpha_i}^{(j)}}(\omega)$. Thus, the result follows from Theorem \ref{thm:02}.

Now assume that $|B|\ge s$, and  $B$ contains $\{1,2,\hdots,s\}.$ Let $a$ be a non-square element in $\F_q^*.$

Take $$\omega=\frac{aydx}{\prod\limits_{i=1}^s(x-\alpha_i)}$$ and the same setting as above. Then, the result follows with the same reasoning as the case $|A|\ge s.$
\end{proof}
\section{Second characterization and construction of $1$-$dim$ hull codes }\label{section:construction2}
In this section, we give another characterization of an algebraic geometry code to have one dimensional hull.
\subsection{Second characterization}
The following lemma gives sufficient conditions for an AG code to have one dimensional hull.
\begin{lem}\label{lem:sufficient2}Let $E$ be a smooth algebraic curve over $\F_q$ with genus $g$ and $G, D = P_1 + P_2 + \cdots+ P_n$ be two divisors over $E$, where $0 < deg(G) < n$. Assume that there exists a differential form $\omega$ such that $(\omega) = G + H - D$ for some divisor $H$ with $Supp(G) \cap Supp(D) = Supp(H) \cap Supp(D) = \emptyset$. Assume further that
\begin{enumerate}
\item $\deg (G\vee H -D)< 0$;
\item $\deg(G\wedge H)=g-i(G\wedge H);$
\item There is a vector ${\bf a}=(a_1,\hdots,a_n)\in (\F_q^*)^n$ with $Res_{P_i}(\omega)=a_i^2$.
\end{enumerate}
Then, ${\bf a}\cdot C_{\cal L}(D, G)$ is a $1$-$dim$ hull code with dimension $k=\deg (G)-g+1\text{ and minimum distance } d\ge n-\deg (G).$
\end{lem}
\begin{proof} Let $c\in C_{\cal L}(D,G)\cap C_{\cal L}(D,H).$ Then, we have that $c=(f(P_1),\hdots,f(P_n))$ $=(g(P_1),\hdots,g(P_n))$ for some $f\in {\cal L}(G)$ and $g\in {\cal L}(H).$ Set $e=f-g.$ If $e=0$, then $f=g\in {\cal L}(G)\cap {\cal L}(H)$. Thus $f=g\in {\cal L}(G\wedge H)$. The fact that $\ell(G\wedge H)=1$ follows from $\deg(G\wedge H)=g-i(G\wedge H)$ and Theorem \ref{thm:Riemann-Roch}.  If $e\not=0,$ then $e\in {\cal L}(G\vee H),$ and since $v_{P_i}(e)\ge 1$, it implies that $e\in {\cal L}(G\vee H-D).$ If $\deg (G\vee H-D)<0$, then $e=0$, a contradiction. Now, take $C=C_{\cal L}(D,G)$ and $C'=C_{\cal L}(D,H)$, and then the fact that ${\bf a}\cdot C_{\cal L}(D,G)$ is $1$-$dim$ hull code follows from Lemma  \ref{lem:dual2} and Lemma \ref{lem:sufficient0}. Finally, the dimension and the minimum distance follow from Thoerem \ref{thm:distance}.
\end{proof}

\begin{lem}\label{lem:new02} Let $E$ be an elliptic curve over $\F_q$ and $G, D = P_1 + P_2 + \cdots+ P_n$ be divisors over $E$, where $0 < deg(G) < n$. Assume that there exists a differential form $\omega$ such that $(\omega) = G + H - D$ for some divisor $H$ with $Supp(G) \cap Supp(D) = Supp(H) \cap Supp(D) = \emptyset$. Assume further that
\begin{enumerate}

\item $\deg (G\wedge H)=1$;
\item There is a vector ${\bf a}=(a_1,\hdots,a_n)\in (\F_q^*)^n$ with $Res_{P_i}(\omega)=a_i^2$.
\end{enumerate}
Then, ${\bf a}\cdot C_{\cal L}(D, G)$ is a $1$-$dim$ hull code with parameters $[n,\deg (G), \ge n-\deg (G)]$, and its dual has parameters $[n,n-\deg (G), \ge \deg (G)]$.
\end{lem}
\begin{proof}  First, note that $G\vee H-D=(\omega)-G\wedge H $. It is well-known that on an elliptic curve the Weil differential $(\omega)$ is principal. 
Since $\deg (G\wedge H)=1$, the divisor $G\wedge H$ is not principal, thus not special and hence $i (G\wedge H)=0$. The rest follows from Lemma \ref{lem:sufficient2}.
\end{proof}

\begin{thm}\label{thm:new02}Let $D = P_1 + P_2 + \cdots+ P_s$ with $P_i=P_i^{(1)}+P_i^{(2)}$ and $\overline{D}$ be such that $\deg \overline{D}=1$,  $\overline{D}\notin E[r+1]$ with $1\le r\le s-1$ and $O, \overline{D}\notin supp(D)$. Set $G=(r+1)O+r\overline{D}$. Assume that there exists a differential form $\omega$ such that $(\omega) = G + H - D$ for some divisor $H$ with $Supp(G) \cap Supp(D) = Supp(H) \cap Supp(D) = \emptyset$, and there exists a vector ${\bf a}=(a_1,\hdots,a_{2s})\in (\F_q^*)^{2s}$ with $Res_{P_i}(\omega)=a_i^2$.
Then, ${\bf a}\cdot C_{\cal L}(D, G)$ is a $1$-$dim$ hull code with parameters $[2s,2r+1, \ge 2s-2r-1]$, and its dual has parameters $[2s,2s-2r-1, \ge 2r+1]$.
\end{thm}

\begin{proof}Take $H=(2s-r-1)O-r\overline{D}$. Under the assumptions in the theorem, we have that $G\wedge H=(r+1)O-r\overline{D}$ has degree one. The result follows from Lemma \ref{lem:new02}.
\end{proof}

\subsection{Second construction from elliptic curves}

We now derive explicit constructions of $1$-$dim$ hull codes from elliptic curves in even characteristic.

\begin{cor}\label{cor:newexplicit1}Let $q$ be even, $n=2s<{N}_{1,b,c}$. Assume that there exists a set $S=\{\alpha_0,\alpha_1,\hdots,\alpha_s\}\subset S_{1,b,c}$. Set
$D=P_{\alpha_1} + P_{\alpha_2}+\cdots+P_{\alpha_s}$ and $G=(r+1)O+rP_{\alpha_0}^{(1)}$ with $1\le r\le s-1$ and $P_{\alpha_0}^{(1)}\notin E[r+1]$, where $P_{\alpha_i}=P^{(1)}_{\alpha_i}+P^{(2)}_{\alpha_i},0\le i\le s.$ Then for some ${\bf a}=(a_{11},a_{12},\hdots, a_{s1},a_{s2})\in (\F_{q}^*)^{n}$, ${\bf a}\cdot C_{\cal L}(D, G)$ is a 
$1$-$dim$ hull code with parameters $[n,2r+1, \ge n-2r-1]$, and its dual has parameters $[n,n-2r-1, \ge 2r+1]$.
\end{cor}

\begin{proof} Take $\omega=\frac{dx}{h(x)}$, where $h(x)=\prod\limits_{i=1}^s(x-\alpha_i)$. Obviously, the residue of the differential form 
$\omega$ at each point $P_{\alpha_i}$ is a nonzero square element in $\F_q$ for $q$ even. Take $a_{i1}^2=Res_{P_{\alpha_i}^{(1)}}(\omega), a_{i2}^2=Res_{P_{\alpha_i}^{(2)}}(\omega)$. Thus, the result follows from Theorem \ref{thm:new02}.
\end{proof}

\begin{cor}\label{cor:newexplicit2} Let $q$ be even, $n=2s< {N}_{1,b,c}$. Assume that there exists a set $S=\{\alpha_0,\alpha_1,\hdots,\alpha_s\}\subset S_{1,b,c}$. Set 
$D=P_{\alpha_1} + P_{\alpha_2}+\cdots+P_{\alpha_s}$ and $G=(r+1)O+rP_{\alpha_0}^{(1)}$ with $1\le r\le s-1$, where $P_{\alpha_i}=P^{(1)}_{\alpha_i}+P^{(2)}_{\alpha_i},0\le i\le s.$ Assume further that $\gcd(r+1,{N}_{1,b,c})=1$. Then for some ${\bf a}=(a_{11},a_{12},\hdots, a_{s1},a_{s2})\in(\F_{q}^*)^{n}$, ${\bf a}\cdot C_{\cal L}(D, G)$ is a 
$1$-$dim$ hull code with parameters $[n,2r+1, \ge n-2r-1]$, and its dual has parameters $[n,n-2r-1, \ge 2r+1]$, where ${\bf a}=(a_{11},a_{12},\hdots a_{s1},a_{s2})$ with $a_{i1}^2=Res_{P_{\alpha_i}^{(1)}}(\omega), a_{i2}^2=Res_{P_{\alpha_i}^{(2)}}(\omega)$.
\end{cor}
\begin{proof}Under the assumption that $\gcd(r+1,{N}_{1,b,c})=1$ by Bezout's theorem, there exist rational integers $a,b$ such that $a(r+1)+b{N}_{1,b,c}=1.$ Then, $(a(r+1)P_{\alpha_0}^{(1)})=P_{\alpha_0}^{(1)}\ominus (b{N}_{1,b,c})P_{\alpha_0}^{(1)}=P_{\alpha_0}^{(1)}\not= O.$ Hence, $P_{\alpha_0}^{(1)}\notin E[r+1]$, and the result follows from Corollary \ref{cor:newexplicit1}.
\end{proof}

%
%
%
%
%
%

\begin{thm}\label{thm:new03} Let $D =P+ P_1 + P_2 + \cdots+ P_s$ with $\deg P =1$, $P_i=P_i^{(1)}+P_i^{(2)}$ and $\overline{D}$ be such that $\deg \overline{D}=1$,  $\overline{D}\notin E[r+1]$ with $1\le r\le s-1$ and $O, \overline{D}\notin supp(D)$. Set $G=(r+1)O+(r-1)\overline{D}$. Assume that there exists a differential form $\omega$ such that $(\omega) = G + H - D$ for some divisor $H$ with $Supp(G) \cap Supp(D) = Supp(H) \cap Supp(D) = \emptyset$, and there exists a vector ${\bf a}=(a,a_1,\hdots,a_{2s})\in (\F_q^*)^{n}$ with $Res_{P}(\omega)=a^2$ and $Res_{P_i}(\omega)=a_i^2$.
Then, ${\bf a}\cdot C_{\cal L}(D, G)$ is a $1$-$dim$ hull code with parameters $[2s,2r+1, \ge 2s-2r-1]$, and its dual has parameters $[2s,2s-2r-1, \ge 2r+1]$.
\end{thm}

\begin{proof} By setting $H=(2s+1-r)O-r\overline{D}$, the result follows with the same reasoning as in the proof of Theorem \ref{thm:new02}. 
\end{proof}

\begin{cor}\label{cor:newexplicit3}Let $q$ be even, $n=2s+1<{N}_{1,b,c}$. Assume that there exists a set $S=\{\alpha_0,\alpha_1,\hdots,\alpha_s\}\subset S_{1,b,c}$. Set
$D=P_{\alpha_0}^{(2)}+P_{\alpha_1} + P_{\alpha_2}+\cdots+P_{\alpha_s}$ and $G=(r+1)O+(r-1)P_{\alpha_0}^{(1)}$ with $1\le r\le s-1$ and $P_{\alpha_0}^{(1)}\notin E[r+1]$, where $P_{\alpha_i}=P^{(1)}_{\alpha_i}+P^{(2)}_{\alpha_i},0\le i\le s.$ Then for some ${\bf a}=(a_{02}, a_{11},a_{12},\hdots a_{s1},a_{s2})\in (\F_q^*)^n$, ${\bf a}\cdot C_{\cal L}(D, G)$ is a 
$1$-$dim$ hull code with parameters $[n,2r+1, \ge n-2r-1]$, and its dual has parameters $[n,n-2r-1, \ge 2r+1]$.
\end{cor}

\begin{proof} Take $\omega=\frac{dx}{h(x)}$, where $h(x)=\prod\limits_{i=0}^s(x-\alpha_i)$ and set $P=P_0^{(2)}$ and $\overline {D}=P_0^{(1)}$ in Theorem \ref{thm:new03}. Obviously, the residue of the differential form 
$\omega$ at each point $P_{\alpha_i}$ is a nonzero square element in $\F_q$ for $q$ even. Take $a_{i1}^2=Res_{P_{\alpha_i}^{(1)}}(\omega), a_{i2}^2=Res_{P_{\alpha_i}^{(2)}}(\omega)$. Thus, the result follows from Theorem \ref{thm:new03}.
\end{proof}

\begin{cor}\label{cor:newexplicit4} Let $q$ be even, $n=2s+1< {N}_{1,b,c}$. Assume that there exists a set $S=\{\alpha_0,\alpha_1,\hdots,\alpha_s\}\subset S_{1,b,c}$. Set 
$D=P_{\alpha_0}^{(2)}+P_{\alpha_1} + P_{\alpha_2}+\cdots+P_{\alpha_s}$ and $G=(r+1)O+(r-1)P_{\alpha_0}^{(1)}$ with $1\le r\le s-1$, where $P_{\alpha_i}=P^{(1)}_{\alpha_i}+P^{(2)}_{\alpha_i},0\le i\le s.$ Assume further that $\gcd(r+1,{N}_{1,b,c})=1$. Then for some ${\bf a}=(a_{02}, a_{11},a_{12},\hdots a_{s1},a_{s2})\in (\F_q^*)^n$, ${\bf a}\cdot C_{\cal L}(D, G)$ is a 
$1$-$dim$ hull code with parameters $[n,2r, \ge n-2r]$, and its dual has parameters $[n,n-2r, \ge 2r]$, where ${\bf a}=(a_{02},a_{11},a_{12},\hdots a_{s1},a_{s2})$ with $a_{i1}^2=Res_{P_{\alpha_i}^{(1)}}(\omega), a_{i2}^2=Res_{P_{\alpha_i}^{(2)}}(\omega)$.
\end{cor}
\begin{proof}The result follows from the same reasoning as that in Corollary \ref{cor:newexplicit2}.
\end{proof}

For elliptic curves over $\F_q$, $q$ odd, we have the following result.
\begin{cor}\label{cor:oddnew1}Let $s|\frac{(q-1)}{2}$, $s$ a square element in $\F_q$ and $n=2s$. Set
$D=P_{\alpha_1} + P_{\alpha_2}+\cdots+P_{\alpha_s}$, where $P_{\alpha_i}=P^{(1)}_{\alpha_i}+P^{(2)}_{\alpha_i},1\le i\le s$ and $G=(r+1)O+rP_{\alpha_0}$ with $1\le r\le s-1$ and $P_{\alpha_0}\notin E[r+1]$. Then for some ${\bf a}=(a_1,\hdots, a_n)\in \left(\F_q^*\right)^n$, ${\bf a}\cdot C_{\cal L}(D, G)$ is a 
$1$-$dim$ hull code with parameters $[n,2r+1, \ge n-2r-1]$, and its dual has parameters $[n,n-2r-1, \ge 2r+1]$.
\end{cor}

\begin{proof}Consider the elliptic curve defined by $${\cal E:}~(y-b)^2=x^3.$$ 
Set $U=\{ \alpha \in \F_q|\alpha^s=1\}$, $h(x)=\prod\limits_{\alpha\in U}(x-\alpha)$, $D=(h)_0$, $P_{\alpha_0}=(0,b)$, $\omega=\frac{dx}{x^s-1}$, $G=(r+1)O+rP_{\alpha_0}$ and $H=D-G+(\omega)=(n-r-1)O-rP_{\alpha_0}$.
We have that $Res_{P_{\alpha_i}}(\omega)=1/(s\alpha_i^{s-1})$ is a square element in $\F_{q}$ for $1\le i\le s$. Thus the result follows by taking $a_i^2=Res_{P_{\alpha_i}}(\omega)$ and from Theorem \ref{thm:new02}.
\end{proof}

\begin{cor}\label{cor:oddnew2}Let $(s+1)|\frac{(q-1)}{2}$, $(s+1)$ a square element in $\F_q$ and $n=2s+1$. Set
$D=P_{\alpha_0}^{(2)}+P_{\alpha_1} + P_{\alpha_2}+\cdots+P_{\alpha_s}$, where $P_{\alpha_i}=P^{(1)}_{\alpha_i}+P^{(2)}_{\alpha_i},0\le i\le s$ and $G=(r+1)O+(r-1)P_{\alpha_0}^{(1)}$ with $1\le r\le s-1$ and $P_{\alpha_0}^{(1)}\notin E[r+1]$. Then for some ${\bf a}=(a_1,\hdots, a_n)\in \left(\F_q^*\right)^n$, ${\bf a}\cdot C_{\cal L}(D, G)$ is a $1$-$dim$ hull code with parameters $[n,2r+1, \ge n-2r-1],$ and its dual has parameters $[n,n-2r-1, \ge 2r+1]$.
\end{cor}

\begin{proof} Take $\omega=\frac{dx}{h(x)}$, where $h(x)=\prod\limits_{i=0}^s(x-\alpha_i)$. The result follows from Theorem \ref{thm:new02} and the same reasoning as in the proof of Corollary \ref{cor:oddnew1}.
\end{proof}

\begin{cor}
Let $q$ be an odd prime power with $q \equiv 1 \pmod 4 $. Then for every $n\equiv 0 \pmod 4$  satisfying  $q \ge  4^{s+4} (s + 4)^2,$ there exist $1$-$dim$ hull codes with parameters $[2s,2r+1,\ge 2(s-r)-1]$ and $[2s,2s-2r-1,\ge 2r+1]$ for $1\le r\le s-1$.
\end{cor}
\begin{proof} The result follows from the same reasoning as that in Corollary \ref{cor:q-odd}.
\end{proof}


\subsection{Second construction from hyper-elliptic curves}

The hyper-elliptic curve over $\mathbb F_{q^2}$ with $q=2^m$ is defined by
\begin{equation}
{\cal C}: ~y^2+y=x^{q+1}.
\label{eq:hyper-elliptic}
\end{equation}
For any $\alpha\in \mathbb F_{q^2}$, there exactly exist two rational points $P_{\alpha}^{(1)}, P_{\alpha}^{(2)}$ with $x$-component $\alpha$.
The set ${\cal C}(\mathbb F_{q^2})$ of all rational points of ${\cal C}$ equal $\{P_{\alpha}^{(1)}| \alpha \in \mathbb F_{q^2}\}\cup \{P_{\alpha}^{(2)}|\alpha \in \mathbb F_{q^2}\} \cup \{ O\}$.

\begin{thm}\label{thm:hyper-elliptic}
Let $q=2^m, m\ge 2$. Assume that $P$ is a point of ${\cal C}$ of degree 1 with $supp(P)\notin D$, $D=\sum\limits_{\alpha\in \F_{q^2}}(P_{\alpha}^{(1)}+P_{\alpha}^{(2)})-P$ and $G=(r+\frac{q}{2}+1)O+r P$ with $\frac{q}{4}\le r \le q^2-\frac{q}{4}-2$. Then, $C_{\mathcal L}(D,G)$ is a $1$-$dim$  hull code with parameters $[2q^2-1,2(r+1), \ge 2(q^2-r-1)-\frac{q}{2}]$, and its dual has parameters $ [2q^2-1, 2(q^2-r-1)-1,\ge 2r+3-\frac{q}{2}]$.
\end{thm}
\begin{proof}
Set $h(x)=\prod\limits_{a\in \F_{q^2}}(x+\alpha)$, $D_0=(h)_0$ and $\omega=\frac{dx}{h(x)}$. Since the curve ${\cal C}$ has genus $g=\frac{q}{2}$ and has $1+2q^2$ rational points, we get
$
(\omega)=2(q^2-1+\frac{q}{2}) O- D_0=2(q^2-1+\frac{q}{2}) - D-P,
$
and for any $\alpha\in \mathbb F_{q^2}$, $\text{Res}_{P_\alpha^{(1)}}(\omega)=\text{Res}_{P_\alpha^{(2)}}(\omega)=1$.

Set $H=D-G+(\omega)$. Then, we get $(G\wedge H)=(r+\frac{q}{2}+1)O-(r+1)P$ with $\deg (G\wedge H)=\frac{q}{2}$, which is non-principal and thus non-special. The result follows from Lemma \ref{lem:sufficient2}.
\end{proof}

\begin{cor}\label{cor:hyper-elliptic0}
Let $q=2^m, m\ge 2$. Assume that $D=P_{\alpha_1}+\cdots+P_{\alpha_s}+P_{\alpha_0}^{(2)} $ with $P_{\alpha_i}=(P_{\alpha_i}^{(1)}+P_{\alpha_i}^{(2)})$ for $0\le i\le s$ and $G=(r+\frac{q}{2}+1)O+r P_{\alpha_0}^{(1)}$ with $\frac{q}{4}\le r \le s-\frac{q}{4}-1$. Then for some ${\bf a}=(a_0,a_1,\hdots, a_{2s})\in \left(\F_q^*\right)^{2s+1}$, ${\bf a}\cdot C_{\mathcal L}(D,G)$ is a $1$-$dim$  hull code with parameters  $[2s+1,2(r+1), \ge 2(s-r)-\frac{q}{2}]$, and its dual has parameters $ [2s+1, 2(s-r)-1,\ge 2r+3-\frac{q}{2}]$.
\end{cor}
\begin{proof}
Let $U=\{\alpha_0,\alpha_1,\hdots,\alpha_s\}$. Set $h(x)=\prod\limits_{a\in U\subset \F_{q^2}}(x+\alpha)$, $D_0=(h)_0$ and $\omega=\frac{dx}{h(x)}$. Since $q$ is even, the residues $\text{Res}_{P_\alpha^{(1)}}(\omega)$ and $\text{Res}_{P_\alpha^{(2)}}(\omega)$ are square elements in $\F_{q^2}.$ The rest follows from the same reasoning as in the proof of Theorem \ref{thm:hyper-elliptic}.
\end{proof}

\begin{cor} \label{cor:hyper-elliptic}
Let $q=2^m,m\ge 3$ and $U_0=\{\alpha\in \F_q| \text{ Tr}(\alpha^5)=0\}.$ Then, there exists a $1$-$dim$ hull code over $\F_q$ with parameters $[2s+1,2(r+1),\ge 2(s-r-1)]$, and it dual has parameters $[2s+1,2(s-r)-1, \ge 2r+1]$ for $1\le s\le |U_0|$ and $1\le r\le s-2 .$
\end{cor}

\begin{proof} 

Consider the hyper-elliptic curve defined by 
\begin{equation}
{\cal C}_0:~y^2 + y = x^5.
\label{eq:hyper-elliptic}
\end{equation}
From Hilbert 90's theorem, we get that $U_0$ is a subset of the solution to (\ref{eq:hyper-elliptic}).
Take two subsets $U,V$ of $U_0$ such that $V\subset U$ with $|V|=s$ and $|U|=s+1$. Set $h(x)=\prod\limits_{\alpha \in U}(x+\alpha)$, $D_0=(h)_0$ and $\omega=\frac{dx}{h(x)}$. Since the curve ${\cal C}_0$ has genus $g=2$, we get
$
(\omega)=(2+2(s+1)) O- D_0,
$
and for any $\alpha\in \mathbb F_{q}$, $\text{Res}_{P_\alpha^{(1)}}(\omega)=\text{Res}_{P_\alpha^{(2)}}(\omega)=\frac{1}{h'(\alpha)}$ are square elements in $\F_q$.
Set $D=D_0-P$, $G=(r+3)O+rP$ and $H=D-G+(\omega)$. The rest follows from the same reasoning as that in Theorem \ref{thm:hyper-elliptic}.

\end{proof}

\subsection{Second construction from Hermitian curves}

The Hermitian curve over $\mathbb F_{q^2}$ with $q=p^m$ is defined by
\begin{equation}
{\cal H}: ~y^q+y=x^{q+1}.
\label{eq:hermitian}
\end{equation}
For any $\alpha\in \mathbb F_{q^2}$, there exactly exist $q$ rational points $P_{\alpha}^{(1)}, P_{\alpha}^{(2)},\hdots, P_{\alpha}^{(q)}$ with $x$-component $\alpha$.
The set ${\cal H}(\mathbb F_{q^2})$ of all rational points of ${\cal H}$ equal $\{P_{\alpha}^{(1)}|\alpha \in \mathbb F_{q^2}\}\cup \{P_{\alpha}^{(2)}| \alpha \in \mathbb F_{q^2}\} \cup \cdots \cup \{P_{\alpha}^{(q)}|\alpha \in \mathbb F_{q^2}\}\cup \{O\}$.

\begin{thm}\label{thm:hermitian-1}
Let $q=p^m\ge 4$. Set $g=\frac{q(q-1)}{2}$. Assume that $P$ is a point of ${\cal H}$ of degree 1, $D=\sum\limits_{\alpha\in \F_{q^2}}(P_{\alpha}^{(1)}+P_{\alpha}^{(2)}+\cdots +P_{\alpha}^{(q)})-P$  and $G=(r+g+1)O+r P$ with $\frac{g-3}{2} < r <  \frac{q^3-g-2}{2} $. Then, $C_{\mathcal L}(D,G)$ is a $1$-$dim$  hull code with parameters $[q^3-1, 2r+2, \ge q^3-\frac{q^2-q}{2}-2r-2]$, and its dual has parameters $[q^3-1, q^3-2r-3,\ge 2r-\frac{q^2-q}{2}+3]$.
\end{thm}
\begin{proof}
Set $h(x)=\prod\limits_{a\in \F_{q^2}}(x-\alpha)$, $D_0=(h)_0$ and $\omega=\frac{dx}{h(x)}$. Since the curve ${\cal H}$ has genus $\frac{q(q-1)}{2}$ and has $1+q^3$ rational points, we get
$
(\omega)=(2g-2+q^3) O- D-P,
$
and for any $\alpha\in \mathbb F_{q^2}$, $\text{Res}_{P_\alpha^{(1)}}(\omega)=\text{Res}_{P_\alpha^{(2)}}(\omega)=-1$.

Set $H=D-G+(\omega)$. Then, $H=(g-r-3+q^3)O-(r+1)P$  and $(G\wedge H)=(r+g+1)O-(r+1)P$ with $\deg (G\wedge H)=\frac{q(q-1)}{2}$, which is non-principal and thus non-special. The result follows from Lemma \ref{lem:sufficient2}.
\end{proof}

\begin{cor}\label{cor:hermitian-2} Let $q_0=p^m, q=q_0^2$ be an odd prime power and $s={q_0}\frac{q_0^2+1}{2}$ and $g=\frac{(q_0-1)^2}{4}$. Then,
there exist $1$-$dim$ hull codes with parameters $[s-1,2r+2,\ge s-2r-g-2]$ and $[s-1,s-2r-3,\ge 2r-g+3]$ for $\frac{g-3}{2} < r < \frac{s-g-2}{2} $.
\end{cor}
\begin{proof}
Consider an algebraic curve defined by
\begin{equation}
{\cal H}_0: y^{q_0}+y=x^{\frac{{q_0}+1}{2}}.
\label{eq:hermitian-2}
\end{equation}
The curve has genus $g=\frac{(q_0-1)^2}{4}$.
Put $$U=\{\alpha\in \F_q|\exists \beta\in \F_q\text{ such that }\beta^{q_0}+\beta=\alpha^{\frac{q_0+1}{2}}\}.$$ 
The set $U$ is the set of $x$-component solutions to the Hermitian curve whose elements are squares in $\F_q.$ There are $\frac{q_0^2+1}{2}$ square elements in $\F_q$, and this gives rise to ${q_0}\frac{q_0^2+1}{2}$ rational places. Write
$$h(x)=\prod\limits_{\alpha\in U}(x-\alpha)\text{ and } \omega=\frac{dx}{h(x)}.$$
Then, $h(x)=x^n-x$, where $n=\frac{q_0^2+1}{2},$ and thus $h'(x)=nx^{n-1}-1$. Since $q$ is a square, we have that $h'(\alpha)=n-1$ is a square for any $\alpha\in U \backslash \{0\}$. Then, the residue $\text{Res}_{P_{\alpha}}(\omega)=\frac{1}{h'(\alpha)}$ is a square for any $\alpha\in U$.
Put $D=(h)_0 -P$ with $\deg P=1$, $G=(r+g+1)O+rP,H=D-G+(\omega)$. Then, we get that $H=(s-r+g-3)O-(r+1)P$ and $G\wedge H=(r+g+1)O-(r+1)P$ with $\deg (G\wedge H)=\frac{(q_0-1)^2}{4}.$ The rest follows from the same reasoning as that in Theorem \ref{thm:hermitian-1}.
\end{proof}

\begin{cor}\label{cor:hermitian-3} Let $q_0=p^m, q=q_0^2$ be an odd prime power. Set $s={q_0}\frac{q_0^2-1}{2}$ and $g=\frac{(q_0-1)^2}{4}$. Then,
 there exist $1$-$dim$ hull codes with parameters $[s-1,2r+2,\ge s-2r-g-2]$ and $[s-1,s-2r-3,\ge 2r-g+3]$ for $\frac{g-3}{2}< r < \frac{s-g-2}{2} $.
\end{cor}
\begin{proof} Consider the same setting as the proof of Corollary \ref{cor:hermitian-2}. Take $U'=U\backslash \{0\},$ and write 
$$h(x)=\prod\limits_{\alpha\in U'}(x-\alpha)\text{ and } \omega=\frac{dx}{h(x)}.$$
The rest follows with the same reasoning as that in Corollary \ref{cor:hermitian-2}.
\end{proof}

\begin{cor}\label{cor:hermitian-4} Let $q_0=p^m$ and $q=q_0^2$. Assume that $s|\frac{q-1}{2}$. Set $N=sq_0-1$ and $g=\frac{q_0(q_0-1)}{2}$. Then, there exists a $1$-$dim$ hull code with parameters $[N,2r+2,\ge N-(2r+1+g) ]$, and its dual has parameters 
$[N,N-2r-2,\ge 2r+3-g]$ for $\frac{g-3}{2} < r <  \frac{N-g-1}{2}$.

\end{cor}
\begin{proof} Consider the Hermitian curve defined by (\ref{eq:hermitian}). Take $U=\{\alpha\in \F_{q}|\alpha^s=1\}$ and $h(x)=x^s-1$. Any element in $U$ is a square in $\F_{q}$, and thus $h'(\alpha)$ is a square for any $\alpha\in U$. Hence, the residues $\text{Res}_{P_\alpha^{(i)}}(\omega)=\frac{1}{h'(\alpha)}$ are square elements in $\F_{q}$ for $1\le i\le q$. Set $D=(h)_0-P$ with $\deg P=1$, $G=(r+1+g)O+rP$, and $\omega=\frac{dx}{h(x)}$. The rest follows from the same reasoning as in the proof of Theorem \ref{thm:hermitian-1}.
\end{proof}

\begin{cor}\label{cor:hermitian-5} Let $q_0=p^m$ and $q=q_0^2$. Assume that $r|2m$. Set $N=p^rq_0-1$ and $g=\frac{q_0(q_0-1)}{2}$.Then, there exists a $1$-$dim$ hull code with parameters $[N,2r+2),\ge N-(2r+1+g) ]$, and its dual has parameters 
$[N,N-2r-2),\ge 2r+3-g]$ for $ \frac{g-3}{2} < r <\frac{N-g-1}{2} $.

\end{cor}
\begin{proof} Consider the Hermitian curve defined by (\ref{eq:hermitian}). Take $U=\{\alpha\in \F_{q}|\alpha^{p^r}=\alpha \}$ and $h(x)=\prod\limits_{\alpha\in U}(x-\alpha)$. Then, $h'(\alpha)$ is a square for any $\alpha\in U$. Hence, the residues $\text{Res}_{P_\alpha^{(i)}}(\omega)=\frac{1}{h'(\alpha)}$ are square elements in $\F_{q}$ for $1\le i\le q_0$. Set $D=(h)_0-P$ with $\deg P=1$, $G=(r+1+g)O+rP$ and $\omega=\frac{dx}{h(x)}$. The rest follows from the same reasoning as in the proof of Theorem \ref{thm:hermitian-1}.
\end{proof}

\begin{cor}$\label{cor}\text{ Let }q_0=p^m$ be an odd prime power  and $q=q_0^2$.
Assume that $n|(q-1)$ and $n$ even. Set $N=(2n+1)q_0-1$ and $g=\frac{q_0(q_0-1)}{2}$. Then, there exists an $1$-$dim$ hull  code with parameters $[N,2r+2,\ge N- (2r+g+1)]$, and its dual has parameters $[N,N-2r-2,\ge 2r+3-g]$ for $ \frac{g-3}{2} < r <\frac{N-g-1}{2} $.
\end{cor}
\begin{proof} Consider the Hermitian curve defined by (\ref{eq:hermitian}).
Let $U_{n}=\{\alpha\in \F_q^*| \alpha^{n}=1\}$. Take $\alpha_1\in \F_q^*$ such that $\alpha_1^{n}-1$ is a nonzero square in $\F_q$ (such $\alpha_1$ does exist). 
Consider the following quadratic equation 
\begin{equation}\label{eq:quadratic}
a^2+b^2=1.
\end{equation}
For any $q$,  (\ref{eq:quadratic}) has $T=(q-1)-4$ solutions, say $(a_1,\pm b_1),\hdots,(a_{\frac{T}{2}},\pm b_{\frac{T}{2}})$, with $(a_i,b_i)\not=(0,\pm 1),(\pm 1,0)$. Take $\alpha_1=\sqrt[n]{a_i^2}$ for some $1\le i\le t,(t< {\frac{T}{2}}).$ Then, we get $1-\alpha_1^n=1-a_i^2=b_i^2$  which are squares in $\F_q^*.$ 

Put $U=U_{n}\cup \alpha_1U_{n}\cup\{0\}$, and write
$$h(x)=\prod\limits_{\alpha\in U}(x-\alpha).$$
Then, we have that $$h'(x)=((n+1)x^{n}-1)(x^{n}-\alpha_1^{n})+nx^{n}(x^{n}-1).$$

We have that $-1, n$ are squares in $\F_q$. Moreover, since $\alpha_1^{n}$ and $(\alpha_1^{n}-1)$ are squares in $\F_q^*$, it implies that $h'(\alpha)$ is a square in $\F_q^*$ for any $\alpha \in U.$ Set $D=(h)_0-P$ with $\deg P=1$, $G=(r+1+g)O+rP$ and $\omega=\frac{dx}{h(x)}$. The rest follows from the same reasoning as in the proof of Theorem \ref{thm:hermitian-1}.
\end{proof}

\section{Application to constructions of EAQECCs}\label{section:application}
In this section, we construct entanglement-assisted quantum error-correcting codes, in short EAQECCs.
Firstly, let us recall some basic notions of quantum codes. Let $\mathbb{C}$ be the complex field and $\mathbb{C}^{q}$ be the $q$-dimensional Hilbert space over $\mathbb{C}$. 
A qubit is a non-zero vector of $\mathbb{C}^{q}$. With a basis $\{|a\rangle : a \in \mathbb{F}_{q}\}$ of $\mathbb{C}^{q}$, any qubit $|v\rangle$ can be expressed as
$|v\rangle=\sum_{a \in \mathbb{F}_{q}}v_{a}|a\rangle,$
where $v_{a} \in \mathbb{C}.$ An $n$-qubit is a joint state of $n$ qubits in the $q^{n}$-dimensional Hilbert space $(\mathbb{C}^{q})^{\bigotimes n}\cong \mathbb{C}^{q^{n}}$. Similarly, an $n$-qubit can be written as
$|\textbf{v}\rangle=\sum_{\textbf{a} \in \mathbb{F}^{n}_{q}}v_{\textbf{a}}|\textbf{a}\rangle,$
where $\{|\textbf{a}\rangle =|a_{1}\rangle\bigotimes|a_{2}\rangle\bigotimes\cdots\bigotimes|a_{n}\rangle| (a_{1}, a_{2}, \ldots, a_{n}) \in \mathbb{F}^{n}_{q}\}$ is a basis of $\mathbb{C}^{q^{n}}$ and $v_{\textbf{a}} \in \mathbb{C}$. 
A quantum code of length $n$ is a subspace of $\mathbb{C}^{q^{n}}$. 
A quantum error $E$ is defined with these following rules. The actions of $X(\textbf{a})$ and $Z(\textbf{b})$ on the basis $|\textbf{v}\rangle \in \mathbb{C}^{q^{n}}$ ($\textbf{v} \in \mathbb{F}^{n}_{q}$) are defined as
\[X(\textbf{a})|\textbf{v}\rangle=|\textbf{v}+\textbf{a}\rangle
\textnormal{ and }
Z(\textbf{b})|\textbf{v}\rangle=\zeta_{p}^{tr(\langle \textbf{v}, \textbf{b} \rangle_{E})}|\textbf{v}\rangle,\]
respectively, where $tr(\cdot)$ is the trace function from $\mathbb{F}_{q}$ to $\mathbb{F}_{p}$. 
For any quantum error $E=\zeta_{p}^{i}X(\textbf{a})Z(\textbf{b})$, where $\zeta_{p}$ is a complex primitive $p$-th root of unity, we define the quantum weight of $E$ by
${\bf wt}_{Q}(E)=\sharp\{i | (a_{i}, b_{i}) \neq (0,0)\}.$
We denote by $[[n,k,d]]_{q}$ a $q$-ary quantum code of length $n$, dimension $k$ and minimum distance $d$.

We use $[[n, k, d; c]]_{q}$ to denote a $q$-ary $[[n, k, d]]_{q}$ quantum code that utilizes $c$ pre-shared entanglement pairs. 
For $c=0$, an $[[n, k, d; c]]_{q}$ EAQECC is equivalent to a quantum stabilizer code \cite{AK01}. 

The following lemma \cite{LA18} gives the constraints among the parameters of an EAQECC.
\begin{prop}(Quantum Singleton Bound )\label{lem4.1}
  For any $[[n, k, d; c]]_{q}$-EAQECC, if $d \leq \frac{n+2}{2}$, we have
  \[ 2(d-1)\le n+c-k .\]
\end{prop}
When the bound meets with equality, the EAQECC is called MDS .

Thanks to Galindo \emph{et al.} \cite{GalHerMatRua}, a $q$-ary EAQECC can be constructed from a classical $q$-ary linear code as follows.
\begin{lem}\cite[Theorem 4]{GalHerMatRua} \label{lem:GalHerMatRua} 
Let $C_1$ and $C_2$ be two linear codes over $\F_q$ with parameters $[n, k_1, d_1]_q$
and $[n, k_2, d_2]_q$ and parity check matrices $H_1$ and $H_2$, respectively. Then, there exists an  EAQECC with parameters $[[n, k_1 + k_2 -n + c, d; c]]$, where 
$d=min\{ d(C_1  \backslash (C_1\cap C_2^\perp)), d(C_2 \backslash (C_1^\perp \cap  C_2))\}$, and
$c = rank(HH^\top)=dim(C_1^\perp)-dim(C_1\perp \cap C_2)$
is the number of required maximally entangled states. Moreover, if there exists a linear code $C$ with parameters $[n,k,d]_q$ such that $dim(C\cap C^\perp)=1$, then there exists an EAQECC with parameters $[[n,k-1,d;n-k-1]]_q.$
\end{lem}

We illustrate constructions of new EAQECCs by the following examples.
\begin{exam}Consider the construction in Corollary \ref{cor:newexplicit2}. For $q=8$ and $w$ primitive element of $\F_8$, the curve defined by 

$$y^2+y=x^3+x+1$$

has 13 rational points.
Take $r=1,  O=(0 : 1 : 0),
    P_0^{(1)}=(w : 0 : 1),
    P_0^{(2)}=(w : 1 : 1)$, $D=\sum\limits_{i=1}^5(P_i^{(1)}+P_i^{(2)})$ and $G=2O+P_{0}^{(1)}$. Then, the code $C_{\cal L}(D,G)$ is an almost $[10, 3, 7]$ linear code over $\F_8$ with the following generator matrix:
$$
\left(
\begin{array}{cccccccccc}
 1  & 0  & 0& w^4 &w^5& w^3& w^2& w^3& w^3& w^4\\
 0  & 1  & 0 &w^4& w^6&   1   &w& w^6& w^4& w^3\\
 0 &  0&   1  & 1& w^3 &w^3 &w^5& w^5 &w^2 &w^2
\end{array}
\right).
$$
It is easy to check that ${\bf a} \cdot C_{\cal L}(D,G)$, where ${\bf a}=(
    w^6,
    w^6,
    w^2,
    w^2,
    w^2,
    w^2,
    w^4,
    w^4,
    w^3,
    w^3)
$,
is a $1$-$dim$ hull code, and from Lemma \ref{lem:GalHerMatRua} it gives rise to a new EAQECC with parameters $[[ 10,2,7;6 ]]_{8}$. Some more new EAQECCs over $\F_{16}$ are given in Table \ref{table:hyper-elliptic}.
\end{exam}

\begin{exam}

Consider the construction in Corollary \ref{cor:hyper-elliptic}. For $q=2^4$ and $w$ primitive element of $\F_{16}$, the curve defined by (\ref{eq:hyper-elliptic}) has $33$ rational points. Take 
$r=1$, $O=(0 : 1 : 0),$
   $P_0^{(1)}=(0 : 0 : 1),$
   $P_0^{(2)}=((0 : 1 : 1),$ $D=P_0^{(2)}+\sum\limits_{i=4}^{15}(P_{\alpha_i}^{(1)}+P_{\alpha_i}^{(2)})$ with $\alpha_i=w^i$ for $4\le i\le 15$ and $G=4O+P_0^{(1)}$.
   
 Then, the code $C_{\cal L}(D,G)$ is a $[25, 4, 20]$ linear code over $\F_{16}$ with the following generator matrix:
  { \footnotesize
$$
\left(
\begin{array}{llllllllllllllllllllllll}
1 0 0 0 0 w^{12} w^{12} w^7 w^7 1 1 w^{11} w^{11} w^8 w^8 w^7 w^7 w^{11} w^{11} w^9 w^9
 w^{12} w^{12} w^2 w^2\\
0 1 0 0 w^{14} w^{14} w^2 w^7 w^2 w^9 w^2 1 w^5 w^3 w w w^{10} w^9 1 w^8 w^{14} w^3
 w^{11} w^{10} w^2\\
0 0 1 0 w^{14} 1 w^6 w^{11} 1 w^7 w^8 w^{14} w^{11} w^8 w^{12} w^5 w^4 w^4 w^3 0 w^6 1
 w^{10} w^9 w^{14}\\
0 0 0 1 1 w^5 w^5 w^{12} w^{12} 1 1 w^{10} w^{10} w^{14} w^{14} w^{11} w^{11} w^5 w^5 w^{11} w^{11}
 w^{10} w^{10} w^3 w^3\\
\end{array}
\right).
$$
}
It is easy to check that ${\bf a} \cdot C_{\cal L}(D,G)$, where
${\bf a}=(
    w^3  
    w  
    w  
    w^7  
    w^7  
    w^8  
    w^8  
    1  
    1  
    w^4  
    w^4  
    1  
    1  
    w^{12}  
    w^{12}  
    w^{10}  
    w^{10}  
    1  
    1  
    w^2  
    w^2  
    w^{10}  
    w^{10}  
    w^{13}  
    w^{13}
)$, 
is a $1$-$dim$ hull code, and from Lemma \ref{lem:GalHerMatRua} it gives rise to a new EAQECC with parameters $[[ 25,3,20;20 ]]_{16}$. Some more new EAQECCs over $\F_{16}$ are given in Table \ref{table:hyper-elliptic}.
\end{exam}

\begin{exam} Consider the construction in Corollary \ref{cor:hermitian-2}. For $q_0=3$ and $w$ primitive element of $\F_9$, the curve defined by (\ref{eq:hermitian-2}) has $16$ rational points. Set 
$
   O=(1: 0 : 0),
   P_0^{(1)}=(0: 0 : 1),
   P_0^{(2)}=(0: w^2 : 1),
   P_0^{(3)}=(0: w^6 : 1)
$

Take $r=1, P=P_0^{(1)}$, $D=P_0^{(2)}+P_0^{(3)})+\sum\limits_{i=1}^4(P_i^{(1)}+P_i^{(2)}+P_i^{(3)})$ and $G=3O+P_{0}^{(1)}$. Then, the code $C_{\cal L}(D,G)$ is an almost $[14, 4, 10]$ linear code over $\F_9$ with the following generator matrix:
$$
\left(
\begin{array}{cccccccccccccc}
 1  & 0&    0  &  0 & w^5 & w^7 &   0&  w^7&  w^7&  w^2 & w^3&  w^2 & w^3&  w^6\\
 0 &   1 &   0 &   0 & w^5 &   2  &  w&  w^2  &  w & w^2 &   1  &  1&    w  &  0\\
 0&    0  &  1  &  0 & w^6 &   2 & w^7 & w^5  &  0&  w^7 &   w&  w^3 & w^6 & w^5\\  0  &  0 &   0 &   1  &  2&  w^3 & w^7 & w^3 & w^7 &   0  &  2  &  1&  w^2 & w^6\\
\end{array}
\right).
$$
It is easy to check that ${\bf a} \cdot C_{\cal L}(D,G)$, where
${\bf a}=(
    1,
    1,
    w^6,
    1,
    1,
    1,
    1,
    1,
    1,
    w^6,
    1,
    1,
    1,
    1),
$
is a $1$-$dim$ hull code and optimal as per \cite{Database}, and from Lemma \ref{lem:GalHerMatRua}  it gives rise to new EAQECCs with parameters $[[ 14, 3, 10; 9 ]]_9$ and $[[ 14, 9, 4; 3 ]]_9$. Some more new EAQECCs over $\F_9,\F_{16},\F_{25}, \F_{49}$ are given in Table \ref{table:hermitian-2}.
\end{exam}

\begin{exam}

Consider the construction in Corollary \ref{cor:hermitian-4}. For $q=3$ and $w$ primitive element of $\F_9$, the curve defined by (\ref{eq:hermitian}) has $28$ rational points. 
Set $O=(0 : 1 : 0),
    P_0^{(1)}=(w^2 : 2 : 1),
    P_0^{(2)}=(w^2 : w : 1),
    P_0^{(3)}=(w^2 : w^3 : 1)
$

Take $r=2$, $D=P_0^{(2)}+P_0^{(3)})+\sum\limits_{i=1}^3(P_i^{(1)}+P_i^{(2)}+P_i^{(3)})$ and $G=6O+P_{0}^{(1)}$. Then, the code $C_{\cal L}(D,G)$ is a $[11, 6, 4]$ linear code over $\F_9$ with the following generator matrix:
$$
\left(
\begin{array}{ccccccccccc}
1& 0& 0& 0 &w^3& 0& 0 &w^7 &w^2& w^5& w^5\\
0& 1& 0& 0 &w^7& 0& 0 &w^3& w^2& 2& 2\\
0& 0& 1& 0& w^2& 0& 0& 0& 1& w& 2\\
0& 0& 0& 1 &w^5& 0& 0& 0& 0 &w^3& 2\\
0& 0& 0& 0& 0& 1& 0 &w^5& w^2& w^6& 2\\
0& 0& 0& 0& 0& 0& 1 &w^2& 0& w^6 &w^7\\
\end{array}
\right).
$$
It is easy to check that ${\bf a} \cdot C_{\cal L}(D,G)$, where
${\bf a}=(
      w^5,
    w^5,
    w^6,
    w^6,
    w^6,
    w^7,
    w^7,
    w^7,
    1,
    1,
    1),
$ 
is a $1$-$dim$ hull code, and from Lemma \ref{lem:GalHerMatRua} it gives rise to an EAQECC with parameters $[[ 11, 5, 4; 4 ]]_9$. This code is not as good as the one given Corollary \ref{cor:hermitian-3}, see Table \ref{table:hermitian-2}.
\end{exam}

\begin{table}[htbp]
\centering
\caption{Some new EAQECCs}
$
\begin{array}{c|c|c}
\text{Corollary \ref{cor:newexplicit2}}&\text{Corollary \ref{cor:newexplicit4}}&\text{Corollary \ref{cor:hyper-elliptic}}\\
\hline
[[6,2,4;2]]_4&&\\

[[10,2,7;6]]_{8}&[[11,1,9;8]]_{8}&[[31,3,26;26]]_{16}\\

[[10,4,5;4]]_{8}&[[11,3,7;6]]_{8}&[[31, 5, 24;24]]_{16}\\

&[[11,5,5;4]]_{8}&[[31, 7, 22;22]]_{16}\\

&[[11,7,3;2]]_{8}&[[31, 9, 20;20]]_{16}\\

[[8,2,5;4]]_{8}&[[9,1,7;6]]_{8}&[[31, 11, 18;18]]_{16}\\

&[[9,3,5;4]]_{8}&[[31, 13, 16;16]]_{16}\\

&[[9,5,3;2]]_{8}&[[31, 15, 14;14]]_{16}\\

[[22,2,19;18]]_{16}&[[23,1,21;20]]_{16}&[[29, 3, 24;24]]_{16}\\

[[22,4,17;16]]_{16}&[[23,3,19;18]]_{16}&[[29, 5, 22;22]]_{16}\\

[[22,6,15;14]]_{16}&[[23,5,17;16]]_{16}&[[29, 9, 18;18]]_{16}\\

[[22,8,13;12]]_{16}&[[23,7,15;14]]_{16}&[[29, 11,16;16]]_{16}\\

[[22,10,11;10]]_{16}&[[23,9,13;12]]_{16}&[[29, 13,14;14]]_{16}\\

[[20,2,17;16]]_{16}&[[23,11,11;10]]_{16}&[[29, 15,12;12]]_{16}\\

[[20,4,15;14]]_{16}&[[21, 1, 19;18]]_{16}&[[27, 3, 22;22]]_{16}\\

[[20,6,13;12]]_{16}&[[21, 3, 17;16]]_{16}&[[27, 5, 20;20]]_{16}\\

[[20,8,11;10]]_{16}&[[21, 5, 15;14]]_{16}&[[27, 7, 18;18]]_{16}\\

[[18,2,15;14]]_{16}&[[21, 7, 13;12]]_{16}&[[27, 9, 16;16]]_{16}\\

[[18,4,13;12]]_{16}&[[21, 9, 11;10]]_{16}&[[27, 11, 14;14]]_{16}\\

[[18,6,11;10]]_{16}&[[21, 11, 9;8]]_{16}&[[27, 13, 12;12]]_{16}\\

[[18,8,9;8]]_{16}&[[19, 1, 17;16]]_{16}&[[25, 3, 20;20]]_{16}\\

&[[19, 3, 15;14]]_{16}&[[25, 5, 18;18]]_{16}\\

&[[19, 5, 13;12]]_{16}&[[25, 7, 16;16]]_{16}\\

&[[19, 7, 11;10]]_{16}&[[25, 9, 14];14]_{16}\\

&[[19, 9, 9;8]]_{16}&[[25, 11, 12;12]]_{16}\\

&[[19, 11, 7;6]]_{16}&[[25, 13, 10;10]]_{16}\\

 \end{array}
$
\label{table:hyper-elliptic}
\end{table}

\begin{table}[htbp]
\centering
\caption{Some new EAQECCs}
$
\begin{array}{c|c}
\text{Corollary \ref{cor:hermitian-2}}&\text{Corollary \ref{cor:hermitian-3}}\\
\hline

[[ 14, 3, 10; 9 ]]_9 & [[ 11, 3, 7; 6 ]]_9\\

[[ 14, 5, 8; 7 ]]_9 & [[ 11, 5, 5; 4 ]]_9\\

[[ 14, 7, 6; 5 ]]_9 & [[ 11, 7, 3; 2 ]]_9\\

[[ 14, 9, 4; 3 ]]_9 & [[ 11, 9, 1; 0 ]]_9\\

%
%
%

[[ 64, 3, 58; 59 ]]_{25}
&
[[ 59, 3, 52; 54 ]]_{25}
\\

[[ 64, 5, 55; 57 ]]_{25}
&
[[ 59, 5, 50; 52 ]]_{25}
\\

[[ 64, 7, 53; 55 ]]_{25}
&
[[ 59, 7, 48; 50 ]]_{25}
\\

[[ 64, 9, 51; 53 ]]_{25}
&
[[ 59, 9, 46; 48 ]]_{25}
\\

[[ 174, 3, 162; 169 ]]_{49}
&
[[ 167, 3, 155; 162 ]]_{49}
\\

[[ 174, 5, 160; 167 ]]_{49}
&
[[ 167, 5, 153; 160 ]]_{49}
\\

[[ 174, 7, 158; 165 ]]_{49}
&
[[ 167, 7, 151; 158 ]]_{49}
\\

[[ 174, 9, 156; 163 ]]_{49}
&
[[ 167, 9, 149; 156 ]]_{49}
\\
 \end{array}
   \label{table:hermitian-2}
$
\end{table}

\section{Conclusion}\label{section:conclusion}
In this paper, we deal with linear codes having one dimensional hull. The hull is defined with respect to Euclidean inner product, and we construct families of $1$-$dim$ hull codes from algebraic geometry codes of genus greater than zero. New EAQECCs are constructed from elliptic curves, hyper-elliptic curves and Hermitian curves. For the future work, with the same spirit, it is worth considering the hull with respect to Hermitian inner product on the one hand, and with higher dimensional hull on the other hand.

%




\begin{thebibliography}{99}











\bibitem{AK01}
A. Ashikhmin and E. Knill, ``Nonbinary quantum stabilizer codes," \emph{IEEE Trans. Inf. Theory}, vol. 47, no. 7, pp. 3065-3072, Nov. 2001.

\bibitem{AssKey} E. F. Assmus, Jr and J. D. Key, ``Affine and projective planes," {\em Discrete Math.} vol. 83, pp. 161--187, 1990.


\bibitem{Mag} W. Bosma and J. Cannon, {\em   Handbook of Magma Functions}, Sydney, 1995.

\bibitem{Bow} G. Bowen, ``Entanglement required in achieving entanglement-assisted channel capacities," {\em Physical Review A,} 66, 052313--1--052313--8 (Nov 2002).

\bibitem{Database} Database of best known linear codes, http://www.codetables.de.
\bibitem{BDH06} T. Brun, I. Devetak and M. H. Hsieh, ``Correcting quantum errors with entanglement,'' \emph{Science}, vol. 314, pp. 436-439, Oct.
2006.
\bibitem{CarGui} C. Carlet and S. Guilley, ``Complementary dual codes for counter-measures to side-channel attacks," In:E.R.
Pinto {\em et al.} (eds.), {\em Coding Theory and Applications,} CIM Series in Mathematical Sciences, vol. 3, pp.
97--105, Springer (2014), {\em J. Adv. Math. Commun.} 10(1), pp.131--150, 2016.
\bibitem{CarGunOzbOzkSol} C. Carlet, C. G\"{u}neri, F. \"{O}zbudak, B. \"{O}zkaya and P. Sol\'e, ``On linear complementary pairs of codes," {\em IEEE
Trans. Inf. Theory,} 64(10), pp. 6583--6589, 2018.
\bibitem{CarLiMes} C. Carlet, C. Li and S. Mesnager, ``Linear codes with small hulls in semi-primitive case," {\em Des. Codes Cryptogr.} https://doi.org/10.1007/s10623-019-00663-4
\bibitem{CarMesTanQiPel18} C. Carlet, S. Mesnager, C. Tang, Y. Qi and R. Pellikaan, `` Linear codes over $\F_q$ are equivalent to LCD codes for $q>3,$" {\em IEEE Trans. Inf. Theory,} 64(4), pp. 3010--3017, 2018.
\bibitem{CarMesTanQi19} C. Carlet C, S. Mesnager, C. Tang and Y. Qi, ``New characterization and parametrization of LCD codes," {\em IEEE Trans. Inf. Theory,} 65(1), pp. 39--49, 2019.
\bibitem{CarMesTanQi18-2} C. Carlet, S. Mesnager, C. Tang and Y. Qi, `` Euclidean and Hermitian LCD MDS codes," {\em Des. Codes Cryptogr.} 86, pp. 2605--2618, 2018.
\bibitem{CarMesTanQi19-2} C. Carlet, S. Mesnager, C. Tang and Y. Qi, ``On $\sigma$-LCD codes," {\em IEEE Trans. Inf. Theory,} 65(3), pp. 1694--1704, 2019.
\bibitem{ChenLiu}B. Chen and H. Liu, ``New constructions of MDS codes with complementary duals," {\em IEEE Trans. Inf. Theory,} 64(8), pp. 5776--5782, 2018.
\bibitem{de} M.A. de Boer, ``Almost MDS codes," {\em Des. Codes Cryptogr.} (1996) 9:143--155. 
\bibitem{DodLan} S.M. Dodunekov and I.N. Landjev, ``Near-MDS codes over some small fields", {\em Discrete Math.} 213 (2000) 55--65.




\bibitem{FangFu} W. Fang and F. Fu, ``New Constructions of MDS Euclidean Self-dual Codes from GRS Codes and Extended GRS Codes," {\em IEEE Trans. Inform. Theory}, vol. 65(9), pp. 5574--5579, 2019.

\bibitem{GalHerMatRua} C. Galindo, F. Hernando, R. Matsumoto, D. Ruano, ``Entanglement-assisted quantum error-correcting codes over arbitrary finite fields," {\em Quantum Information Processing,} 18(4), 116 (Apr 2019).

\bibitem{Goppa}V. D. Goppa, ``Algebraico-geometric codes," {\em Math. USSR-lvz.} 21(1) (1983) 75-91.
\bibitem{GraGul} M. Grassl and T. A. Gulliver, ``On Self-Dual MDS Codes" {\em ISIT 2008}, Toronto, Canada, July 6 --11, 2008.
\bibitem{Gue} K. Guenda, ``New MDS self-dual codes over finite fields,"  {\em Des. Codes Cryptogr.} (2012) 62:31--42
\bibitem{GGJT18}
K. Guenda, T.A. Gulliver, S. Jitman and S. Thipworawimon, ``Linear $\ell$-intersection pairs of codes and their applications,'' \emph{Des. Codes
Cryptogr.}, 2019. https://doi.org/10.1007/s10623-019-00676-z

\bibitem{Jin} L. Jin, ``Construction of MDS codes with complementary duals," {\em IEEE Trans. Inf. Theory,} 63(5), pp. 2843--2847, 2017.
\bibitem{JinBee}L. Jin and P. Beelen, ``Explicit MDS Codes With Complementary Duals," {\em IEEE Trans. Inf. Theory,} 64(11), pp. 7188 --7193, 2018.

\bibitem{JinXin} L. Jin and C. Xing, ``New MDS self-dual codes from generalized Reed-Solomon codes," {\em IEEE Trans. Inform. Theory}, vol. 63(3) , pp. 1434 --1438, 2017.

\bibitem{LA18}
 C. Y. Lai and A. Ashikhmin, ``Linear programming bounds for entanglement-assisted quantum error correcting codes by split weight enumerators,'' \emph{IEEE Trans. Inf. Theory}, vol.
64, no. 1, pp. 622-639, Jan. 2018.

\bibitem{LB13}
C.-Y. Lai and T. A. Brun, ``Entanglement increases the error-correcting
ability of quantum error-correcting codes,'' \emph{Phys. Rev. A, Gen. Phys.},
vol. 88, p. 012320, Jul. 2013.

\bibitem{Leon82}J. Leon, ``Computing automorphism groups of error-correcting codes," {\em IEEE Trans. Inf. Theory,} 28(3), pp. 496--511, 1982.
\bibitem{Leon91} J. Leon, ``Permutation group algorithms based on partition, I: Theory and algorithms," {\em J. Symb. Comput.} 12, pp. 533--583, 1991.
\bibitem{LiDingLi} C. Li, C. Ding and S. Li, ``LCD cyclic codes over finite fields," {\em IEEE Trans. Inf. Theory,} 63(7), pp. 4344--4356, 2017.
\bibitem{LiLiDingLiu}S. Li, C. Li, C. Ding and H. Liu, ``Two Families of LCD BCH codes," {\em IEEE Trans. Inf. Theory,} 63(9), pp. 5699--5717, 2017.
\bibitem{LiZeng} C. Li and P. Zeng, ``Constructions of linear codes with one-dimensional hull," {\em IEEETrans. Inf. Theory,} 65 (3), pp. 1668--1676, 2019.

\bibitem{LCC}
G. Luo, X. Cao and X. Chen, ``MDS codes with hulls of arbitrary dimensions and
their quantum error correction,'' \emph{IEEE Trans. Inf. Theory}, vol. 65, no. 5, pp. 2944-2952, May 2019.
\bibitem{Massey}J. L. Massey, ``Linear codes with complementary duals," {\em Discret. Math.}106(107),337--342 (1992).
\bibitem{Massey} J. Massey, ``Some applications of coding theory in cryptography," in
{\em Proc. 4th IMA Conf. Cryptogr. Coding,} 1995, pp. 33--47.
\bibitem{MesTanQi}S. Mesnager, C. Tang and Y. Qi, ``Complementary dual algebraic geometry codes," {\em IEEE Trans. Inf. Theory,} 64(4), pp. 2390--2397, 2018.

\bibitem{QCM} L. Qian, X. Cao and S. Mesnager, ``Linear codes with one-dimensional hull associated with Gaussian sums", {\em Cryptogr. Commun.} (2020). https://doi.org/10.1007/s12095-020-00462-y



\bibitem{SangJitLingUdom}E. Sangwisut, S. Jitman, S. Ling, and P. Udomkavanich, ``Hulls of cyclic and negacyclic codes over finite fields," {\em Finite Fields Appl.} 33, pp. 232--257, 2015.
\bibitem{Sendrier97}N. Sendrier, ``On the dimension of the hull," {\em SIAM J. Discret. Math.} 10(2), pp. 282--293, 1997.

\bibitem{Sendrier00}N. Sendrier, ``Finding the permutation between equivalent codes: the support splitting algorithm," {\em IEEE Trans. Inf. Theory,} 46(4), pp. 1193--1203, 2000.

\bibitem{ShiYueYan} X. Shi, Q. Yue and S. Yang, ``New LCD MDS codes constructed from generalized Reed-Solomon codes," {\em J. Algebra Appl.} 18, 1950150, 2018.

\bibitem{SHB11}
J. Shin, J. Heo and T. A. Brun, ``Entanglement-assisted codeword
stabilized quantum codes,'' \emph{Phys. Rev. A, Gen. Phys.}, vol. 84, p. 062321, Dec. 2011.

\bibitem{Ske}G. Skersys, ``The average dimension of the hull of cyclic codes," {\em Discret. Appl. Math.} 128(1), pp. 275--292, 2003.

\bibitem{Sok} L. Sok, ``Explicit constructions of MDS self-dual codes," {\it IEEE Transactions on Information Theory,} vol. 66(6), pp. 3603--3615, 2020, doi:10.1109/TIT.2019.2954877
\bibitem{Sok1D} L. Sok, ``MDS linear codes with one dimensional hull," https://arxiv.org/pdf/2012.11247.pdf
\bibitem {Stich} H. Stichtenoth, ``Algebraic function fields and codes," Springer, 2008.
\bibitem{TongWang} H. Tong and X. Wang, ``New MDS Euclidean and Hermitian self-dual codes over finite fields," {\em Adv. in Pure Math.} , vol. 7, pp. 325--333, May. 2017.

\bibitem{WilBru} M.M. Wilde and T.A. Brun, ``Optimal entanglement formulas for entanglement-assisted quantum coding," {\em Physical Review A,} 77(6), 064302--1--064302--4 (Jun 2008).

\bibitem{Yan} H. Yan, ``A note on the constructions of MDS self-dual codes," {\em Cryptogr. Commun.,} (2019) 11:259--268,
https://doi.org/10.1007/s12095-018-0288-3.


\bibitem{YanLiuLiYang}H. Yan, H. Liu, C. Li and S. Yang, ``Parameters of LCD BCH codes with two lengths," {\em Adv. Math. Commun.}12(3), pp. 579--594, 2018.
\bibitem{YangMassey} X. Yang and J. L. Massey, ``The condition for a cyclic code to have a complementary dual," {\em Discret. Math.} 126(1--3), pp. 391--393, 1994.

\end{thebibliography}
\end{document}